\documentclass[journal]{IEEEtran}
\usepackage{cite}
\usepackage{multirow}
\usepackage{comment}

\ifCLASSINFOpdf
    \usepackage[dvips]{graphicx}
\else
    \usepackage{graphicx}
\fi

\usepackage{amsmath}

\usepackage{algorithm}
\usepackage{algpseudocode}

\usepackage{array}

\usepackage[utf8]{inputenc}
\usepackage{amssymb}
\usepackage{amsthm}
\usepackage{nccmath}
\usepackage{mathtools}

\newtheorem{theorem}{Theorem}
\newtheorem{lemma}{Lemma}

\graphicspath{{./figures/}}

\newcommand{\EX}[1]{\mathbb{E}\left[ #1 \right]}
\newcommand{\CEX}[2]{\EX{#1 \middle| #2}}
\newcommand{\PR}[1]{\mathbb{P}\left[ #1 \right]}
\newcommand{\CPR}[2]{\PR{#1 \mid #2}}

\newcommand{\as}{\quad \textit{a.s.}}
\newcommand{\IND}[1]{1_{\{#1\}}}
\newcommand{\one}{\mathbf{1}}
\newcommand{\norm}[1]{\left\lVert#1\right\rVert}

\newcommand{\E}{\mathrm{e}}

\newcommand{\SRADIUS}[1]{\rho\left(#1 \right)}

\newcommand{\algname}{\textit{RoamingToken}}

\newcommand{\test}{s}

\newcommand{\cmt}[1]{\ifCLASSOPTIONdraftcls\footnote{#1}\fi}

\DeclareMathOperator{\var}{var}
\DeclareMathOperator{\eig}{eig}
\DeclareMathOperator{\degree}{degree}

\newtheorem{assumption}{Assumption}

\begin{document}

\title{Distributed Estimation Via a Roaming Token}

\author{Lucas~Balthazar,
        João Xavier,
        and~Bruno~Sinopoli~%<-this % stops a space
\thanks{Lucas Balthazar is with the Institute for Systems and Robotics (ISR), Instituto Superior Técnico (IST), Technical University of Lisbon, Lisbon, Portugal. He is also with the Department of Electrical and Computer Engineering, Carnegie Mellon University, Pittsburgh, PA, USA (e-mail: lbalthaz@andrew.cmu.edu).}% <-this % stops a space
\thanks{João Xavier is with the Institute for Systems and Robotics (ISR), Instituto Superior Técnico (IST), Technical University of Lisbon, Lisbon, Portugal (e-mail: jxavier@isr.tecnico.ulisboa.pt).}% <-this % stops a space
\thanks{Bruno Sinopoli is with the Department of Electrical and Computer Engineering, Carnegie Mellon University, Pittsburgh, PA, USA (e-mail: brunos@ece.cmu.edu).}% <-this % stops a space
\thanks{The work of L. Balthazar, J. Xavier and B. Sinopoli was work was supported in part by Fundação para a Ciência e Tecnologia (the Portuguese Foundation for Science and Technology) through the Carnegie Mellon Portugal Program under Projects UID/EEA/50009/2013 and SFRH/BD/52533/2014.}}

%in part by Fundação para a Ciência e Tecnologia under Project UID/EEA/50009/2013. Email: jxavier@isr.tecnico.ulisboa.pt.

%\markboth{IEEE TRANSACTIONS ON SIGNAL PROCESSING}%
{}%{Shell \MakeLowercase{\textit{et al.}}: Bare Demo of IEEEtran.cls for IEEE Journals}

\maketitle

\begin{abstract}
We present an algorithm for the problem of linear distributed estimation of a parameter in a network where a set of agents are successively taking measurements. The approach considers a roaming token in a network that carries the estimate, and jumps from one agent to another in its vicinity according to the probabilities of a Markov chain. When the token is at an agent it records the agent's local information. We analyze the proposed algorithm and show that it is consistent and asymptotically optimal, in the sense that its mean-square-error (MSE) rate of decay approaches the centralized one as the number of iterations increases. We show these results for a scenario where the network changes over time, and we consider two different set of assumptions on the network instantiations: they \textit{i.i.d.} and connected on the average, or that they are deterministic and strongly connected for every finite time window of a fixed size. Simulations show our algorithm is competitive with \textit{consensus+innovations} type algorithms, achieving a smaller MSE at each iteration in all considered scenarios. 
\end{abstract}

\section{Introduction}

\IEEEPARstart{T}{his} paper considers the problem of distributed estimation. In a typical setting, many agents are deployed in a region and are interested in estimating a parameter from their private measurements. These agents could represent a wireless sensor network, a fleet of mobile robots, or multiple devices connected in a internet-of-things setup. Due to physical constraints agents are restricted to communicating only with a subset of other agents, and can be seen as forming a network of interconnected entities. Further, the agents typically have computational capabilities and can be used for more than simply measuring and relaying information.  The main feature of \textit{peer-to-peer} algorithms of this kind is the absence of a central node, and typically the agents' capabilities are exploited so as to perform the estimation procedure in parallel while using the network to share their intermediate computations.

% The task is then to exploit these capabilities and determine the value of a parameter from the agents' datasets.

%and normally the network operate in a symmetric way, with each agent performing the same kind of computations. 

%While canonical estimation theory would require access to all measurements in order to perform computations, when the measurements are distributed across agents in a network it becomes advantageous to use their processing abilities for intermediate computations, and in a typical setting the processing effort is distributed evenly across the network.

Early work considered networks of agents where each makes only one measurement. In the standard case, the problem reduces to the one of computing the average of the measurements, and the usual approach has been to use consensus algorithms to perform the measurement fusion, for example as in \cite{xiao2005scheme}. Later on the case where agents are continuously making measurements has been considered. Examples of these include the \textit{consensus+innovations} algorithms, as presented and studied in \cite{kar2011convergence,kar2013distributed,kar2013consensus+}, and the diffusion least-mean square algorithms, as in \cite{lopes2008diffusion,cattivelli2010diffusion}. Both of these include a consensus-like step, where agents cooperate with their neighbors by sharing their current estimate and performing a convex combination of the nearby estimates, and a innovation step where agents use exclusively their local information to update their current estimate. 

In the \textit{consensus+innovations} type algorithms the two steps occur at different timescales, with different diminishing parameters. The algorithm is shown to be asymptotically efficient, in that it asymptotically has the same mean square error decay rate as the estimator obtained by an oracle having access to all measurements at all time instants. In the diffusion algorithms, constant step sizes are used instead, which yield worse asymptotic performance but gives the network the ability to react to changing statistics of the parameters.

Other types of algorithms proposed for distributed estimation include those of the incremental type, for example as in \cite{lopes2007incremental, ram2007stochastic, li2010distributed}. In these references, only one agent is active at each time instant. That agent will take a measurement of the environment, update the estimate, and send the estimate to one of its neighbours, which will then proceed in the same fashion. Thus there is only one estimate at each instant of time, localized in the agent that is currently active. The algorithm requires an initial setup time, where the agents communicate between themselves so as to find a cycle in the network, which will determine the path the estimate makes as it is being passed between the agents. 

Another approach of the incremental type is that of \cite{johansson2007simple,johansson2009randomized, ram2009incremental}. In \cite{johansson2009randomized} it is presented an iterative algorithm to incrementally update a subgradient method for distributed optimization. In this case, the variable of interest does not follow a fixed path, but instead is passed from agent to agent in a randomized way, according to probabilities that follow a Markov chain. This was later extended in \cite{ram2009incremental} to consider subgradient updates that are affected by random noise, and with this extension the algorithm can be applied to an estimation setting.  

In this paper we propose a different type of distributed algorithm for linear estimation. We consider a setting similar to \cite{kar2011convergence} in terms of the general measurement model of the agents, and in particular we consider that at each time instant, every agent is active taking a measurement of the environment. However, we do not use a consensus step to fuse the different agents estimate, and instead we follow an approach that is closer to \cite{johansson2009randomized, ram2009incremental}, where a token travels the network carrying the estimate. In the same way as it is done in \cite{johansson2009randomized}, the token is passed to a neighbor according to the probabilities of a Markov chain, and communication happens in a directed fashion. When the token is at an agent, the estimate is updated with that agent local information. A novelty of our algorithm is that contrary to a typical distributed estimation algorithms, the agents will not keep an estimate of the quantity of interest that is updated at each iteration, and instead they keep a variable that is updated at each iteration with only local information, as they wait for the token to arrive. Thus our algorithm is neither an incremental algorithm nor a consensus-type algorithm. 

In our algorithm, the most up-do-date estimate at a particular time instant is located in the agent that carries the token at that instant. As a possible intentional application, consider for example a situation where the agents represent a sensor network that is tasked with performing an action depending on the parameter, and any agent can perform such action\cmt{Need to edit this paragraph in view of comment: 'I liked this justification but find it too vague. Refine this example; go to specifics; make up a plausible scenario anchored in a usual framework (e.g., cognitive radio?)'}. In this case it is enough that one agent has at a given time instant the estimated value of the parameter, as our algorithm guarantees. Alternatively, each agent can save the estimate from the last time it was visited by the token, and use that as its current estimate. 

\textit{Contributions:} We present a novel algorithm for distributed linear estimation and show that it is asymptotically optimal, in the sense that as the number of iterations grows to infinity, the decay rate of the mean square error (MSE) of our estimate approaches the optimal rate of decay. We show this result considering a directed communication model, and considering two different assumptions on the network connections instantiations: (I) we consider they are \textit{i.i.d} and strongly connected on the average, or (II) that they are deterministic with the property that the union of the networks at each time window of a fixed size is strongly connected. We believe this is the first distributed algorithm for estimation that guarantees an optimal asymptotic rate of decay for the MSE under scenario (II).  Under scenario (I), we believe it is the first algorithm with directed communication that guarantees this, and we also show via simulations that our algorithm can outperform the \textit{consensus+innovations} algorithm in a number of test cases. 

The rest of this paper is organized as follows. In section \ref{sec:model} we present our model and the underlying assumptions. In section \ref{sec:particle} we present our algorithm and a possible implementation of it. Section \ref{sec:proof} is devoted to showing the consistency and asymptotic optimality of the proposed algorithm when we consider that the network instantiations are \textit{i.i.d}. Section \ref{sec:simulation} provides some numerical experiments and comparison with a \textit{consensus+innovations} type algorithm. Section \ref{sec:networkB} generalizes the theoretical results previously obtained and show how they can be applied to a more general setting, and we focus on one where the network instantiations are deterministic and strongly connected at each finite time window of a fixed size.

%Section \ref{sec:networkB} generalizes the theoretical results previously obtained and show how they can be applied to a setting where the network instantiations are deterministic

\textit{Notation:} Capital letters are used to denote matrices, sets, and events. For a matrix $A$, $A_{ij}$ is used to denote the entry in the $i$-th row and $j$-th, $A^\top$ is the transpose, $\SRADIUS{A}$ is the spectral radius. $I$, $e_i$, $\one$, are respectively, the identity matrix, the $i$-th coordinate vector and a vector with all entries equal to $1$, with size given by context. When a specific size is meant a subscript is added, thus $I_k \in \mathbb{R}^{k \times k}$ and $\one_k \in \mathbb{R}^k$. We use $\IND{E}$ to denote the indicator random variable for event $E$, thus $\IND{E} = 1$ if and only if the event $E$ occurred. $\norm{v},\norm{A}$ are respectively the $2-$norm of vector $v$ and induced $2-$norm of matrix $A$. For a set $S$, $|S|$ denotes its cardinality. For a random variable $x$ we denote by $\EX{x}$ its expected value.

 

\section{Model and assumptions}
\label{sec:model}

We consider a set of agents $\mathcal{V}$, with $|\mathcal{V}| = n$. The variable $t$ indicates time, which is measured in discrete increments, with the agents synchronized at each time step. At every $t = 0,1,\ldots$ agent $i$ takes a measurement 
\begin{equation}
\label{eq:measurement}
y_i(t) = H_i \theta + w_i(t) 
\end{equation}
where $H_i$ is the $i$-th agent observation matrix, $\theta$ is the parameter we are trying to estimate, and $w_i(t)$ is the noise process that affects the measurement. We note that $y_i(t)$, $\theta$ and $w_i(t)$ are vectors of sizes consistent with equation \eqref{eq:measurement}. We define $y(t)^\top = [y_1(t)^\top \ldots y_n(t)^\top ]$, $H^\top = [H_1^\top \ldots H_n^\top ]$, $w(t)^\top = [w_1(t)^\top \ldots w_n(t)^\top ]$. Then the global measurement model at time $t$ can be written as
\begin{equation}
y(t) = H \theta + w(t).
\end{equation}
We consider two assumptions on the noise sequence and observation matrix $H$. 
\begin{assumption}
\label{asp:noise}
The noise sequence $(w(t))_{t\geq 0}$ is independent, spatially uncorrelated $\EX{w_k(t) w_l(t)^\top} = 0,$ for $l \neq k$, zero mean $\EX{w(t)} = 0$, and with covariance $\EX{w_i(t)} = C_i$, for all $t$. 
\end{assumption}
\begin{assumption}
\label{asp:invert}
The matrix $H^\top H$ is invertible.
\end{assumption}
\noindent We note that time independence is usually assumed in most works in distributed estimation, in particular in all works mentioned in the introduction\cmt{Still need to comment with regards to spatially uncorrelated. It's a bit messy because consensus+innovations can deal with spatially correlated noise, but it requires optimal update matrices, and so far there is no way to compute it in a distributed setup. When there are solutions to compute (distributively) the optimal update matrices, they assume the noise is uncorrelated, \textit{e.g.} in \cite{kar2013distributed}. When the noise is correlated between different sensors, then consensus+innovations guarantees that the estimate is asymptotically normal and evolves as for example $\frac{G}{t}$ where $G$ is (without global knowledge) a non-optimal matrix. Thus I can't say that there are no works that assume the noise to be spatially correlated, but at the same time, such works have no viable distributed implementation.}. Assumption \ref{asp:invert} is necessary to guarantee that even the centralized problem of estimation is well posed. 

A classic problem in estimation is the determination of $\theta$ after $t$ measurements have been obtained. Under assumptions \ref{asp:noise} and \ref{asp:invert} the best linear estimator for $\theta$ can be found as
\begin{equation}
\label{eq:central_estimator}
\hat{\theta}_c(t) = \left(\sum_{i=1}^n H_i^\top C_i^{-1} H_i\right)^{-1} \sum_{i=1}^n H_i^\top C_i^{-1} \overline{y}_i(t)
\end{equation}
where $\overline{y}_i(t) = \frac{1}{t}\sum_{s=0}^t y_i(s)$. If the noise follows a Gaussian distribution, this estimator is an efficient estimator of $\theta$ as it achieves the Cramér-Rao lower bound for any $t$, and thus it is the minimum variance unbiased estimator. For a generic distribution, it can be shown to be the linear estimator for $\theta$ with minimum variance among all linear estimators. These are known results from estimation theory and we omit the details here for conciseness (see for example \cite{kay1993fundamentals}). Because of the properties listed, estimator $\hat{\theta}_c$ is a desirable estimator in a wide number of situations, and thus its extension to a distributed setting is of great interest.

The estimator \eqref{eq:central_estimator} we call the central estimator of $\theta$, as it is a desirable estimator which depends at time $t$ on all measurements of all agents. Its covariance is given by $\var(\hat{\theta}_c(t)) = \frac{\Sigma_c^{-1}}{t}$, where $\Sigma_c$ is the Fisher information matrix for Gaussian noise $$\Sigma_c = \sum_{i=1}^n H_i^\top C_i^{-1} H_i.$$ In a distributed setting, agents do not have instantaneous access to all the measurements obtained by the network, and so a distributed estimator will generally have a higher variance than the central estimator at any time instant $t$. A question that can be asked is if distributed estimators can achieve at least asymptotically the mean-square-error rate of decay of the central estimator, that is if for our distributed estimate $\hat{\theta}(t)$ we can guarantee that
\begin{equation}
    \label{eq:asymptoticallyoptimal}
    \lim_{t \rightarrow +\infty} t \EX{ (\hat{\theta}(t) - \theta)(\hat{\theta}(t) - \theta)^\top} = \Sigma_c^{-1}.
\end{equation}
For the sake of not being verbose we will call an estimator that verifies \eqref{eq:asymptoticallyoptimal} asymptotically optimal. It is well known that distributed estimators of the \textit{consensus+innovations} type \cite{kar2011convergence,kar2013distributed,kar2013consensus+} are asymptotically optimal \cmt{I believe they are the only distributed algorithms for estimation that have been shown to be asymptotically optimal?}. One of the goals of this paper is to show that our proposed algorithm also guarantees this, considering the setting where the network connections vary with time, as long as they satisfy some mild assumptions that are also typically encountered in the literature. 

\paragraph*{A simple distributed algorithm} In order to motivate our algorithm, we present first a simple procedure that allows for the estimation of $\theta$ in a distributed setting when we consider a connected and static network. We assume the following.

\begin{assumption}
\label{asp:distributed}
Each agent knows its measurement matrix $H_i$ and noise covariance matrix $C_i$.
\end{assumption}

If assumption \ref{asp:distributed} holds, we can task each agent with computing the local variable
\begin{equation}
x_i(t) = H_i^\top C_i^{-1} \overline{y}_i(t)
\end{equation}
and keep it updated at every time instant. Looking at the expression in \eqref{eq:central_estimator} we could reproduce the central estimator by gathering the different $x_i(t)$ and left multiplying by the constant matrix $\Sigma_c^{-1}$. We suppose the network has been setup so as to form a cycle
%, as illustrated in figure \ref{fig:simpleprocedure},
and for the moment we also assume all the agents know the value of the constant matrix $\Sigma_c$. Then, a procedure to estimate $\theta$ can be as follows. At time $t=0$ agent $1$ sends $x_1(0)$ to agent $2$, which at time $t = 1$ sends $x_1(0)+x_2(1)$ to agent $3$, and so on. If we let $d(t)$ denote the quantity that is sent at time $t$, agent $n$ will compute at time $n-1$ the estimate $\hat{\theta}(n-1) = \Sigma_c^{-1} d(n-1)$, where $d(n-1) = x_1(0) + x_2(1) + \ldots + x_n(n-1)$. Continuing, at time $n$, agent $1$ receives $d(n-1)$ from agent $n$, updates it as $d(n) = d(n-1) - x_1(0) + x_1(n)$, and computes a new estimate. Thus, using this procedure, if we let $p(t)$ denote the agent where $d(t)$ is, so that $p(0) = 1, p(1) = 2, \ldots$, and we define
\begin{equation}
\label{eq:definitionTau}
\tau_i(t) = \sup\{s \leq t : p(s) = i\}
\end{equation}
we can then write 
\begin{equation}
\hat{\theta}(t) = \Sigma_c^{-1} \sum_{i=1}^n x_i(\tau_i(t)).
\end{equation}
We note, $t - n \leq \tau_i(t) \leq t$, and thus the estimate $\hat{\theta}(t)$ is not exactly the one in \eqref{eq:central_estimator}, but it gets close to it as the number of measurements increase, in a specific sense that we will make precise later.  

The procedure just described has some drawbacks when the distributed structure of the network is considered:
\begin{enumerate}
\item It requires us to establish a cycle in the network, which can be unfeasible for large networks. 
\item If we consider a scenario where communication links can change over time, requiring the communication structure to follow a cycle can introduce an excessive delay in the estimation procedure. 
\end{enumerate}
In this paper, we will present a novel algorithm based on this procedure that does not require a cycle to be formed, and takes into account the changing topology of the network. We will do this by letting the quantity $d(t)$ roam in the network, so that $p(t)$ has a random nature to it, which we denote by the token $p(t)$. As the token $p(t)$ visits an agent, it updates the quantity $d(t)$ with the agent's updated local variable $x_i(t)$, and is then passed to another one, according to the available communications at that time instant. Before making the description of our algorithm more precise, we introduce some assumptions about the network in which the agents communicate. 

\paragraph*{Network model} We will consider a dynamic setting, where the network connections change over time. For this purpose, we introduce a sequence of matrices $(A(t))_{t \geq 0}$ where each $A(t)$ is an adjacency matrix for a directed graph with $n$ nodes, $A(t) \in \mathcal{A}^n$, where
\begin{equation*}
\mathcal{A}^n = \{B \in \{0,1\}^{n \times n}: B_{ii} = 0,~ i = 1,\ldots,n\}.
\end{equation*}
The matrix $A(t)$ represents which communication links can be used at time $t$, and $(A(t))_{ij} = 1$ if and only if node $i$ can send information to node $j$ at time instant $t$. We allow for directed communication, thus $A(t)$ is possibly asymmetric and the row $i$ of the adjacency matrix corresponds to the outward edges from node $i$ that exist on the graph. For $A \in \mathcal{A}^n$, we denote $\mathcal{G}(A) = (\mathcal{V},\mathcal{E}(A))$ the corresponding graph, with $\mathcal{E}(A)$ the corresponding set of edges, and where $\mathcal{G}$ is such that $(i,j) \in \mathcal{E}(A)$ if and only if $A_{ij} > 0$. We consider first the following assumption on the sequence $A(t)$.
\begin{assumption}
\label{asp:networkA}
The sequence $(A(t))_{t \geq 0}$ is random, \textit{i.i.d.}, and independent of $(w(t))_{t \geq 0}$. Further, the graph $\mathcal{G}(\EX{A(t)})$ is strongly connected. 
\end{assumption}
\noindent
We recall that a graph is strongly connected if and only if for any $i,j \in \mathcal{V}$ there is a path connecting $i$ to $j$. Under assumption \ref{asp:networkA} we can let $\overline{\mathcal{A}} = \{A \in \mathcal{A}^n: \PR{A(t) = A} > 0\}$, and then we have
\begin{align*}
\mathcal{G}(\EX{A(t)}) & = \mathcal{G}\left( \sum_{A \in \overline{\mathcal{A}}} A \PR{A(t) = A} \right)\\ & = \left(\mathcal{V}, \bigcup_{A \in \overline{\mathcal{A}}} \mathcal{E}(A) \right), 
\end{align*} 
so that assumption \ref{asp:networkA} is equivalent to assuming that the sequence is \textit{i.i.d} and that the union of all graphs that have some probability of appearing in the sequence is strongly connected. We note that assumption \ref{asp:networkA} together with assumption \ref{asp:invert} are the same as the distributed observability requirement that is shown to be necessary to obtain asymptotic efficiency in \textit{consensus+innovations} algorithms \cite{kar2011convergence, kar2013consensus+, kar2013distributed}. For the sake of not overburdening the paper, we focus first exclusively on a network in which assumption \ref{asp:networkA} holds. In section \ref{sec:networkB} we explain how the results developed can be applied to the setting where instead of being random, the network instantiations are assumed to be deterministic with the property that they are strongly connected in each finite time window of a certain size.  

\paragraph*{Roaming token} We now make precise the nature of the roaming token $(p(t))_{t \geq 0}$. We will let $p(t) \in \mathcal{V}$, and $p(t)$ corresponds to the agent, or node, which currently has access to the gathered measurements variable $d(t)$. We will let $p(t+1)$ be chosen randomly, between the available connections at agent $p(t)$ at time $t$. Specifically, we will let the token move as a Markov chain, which at time $t$ has a transition matrix that is consistent with the available communication links of the graph. To make this precise, we consider a function $Q(.): \mathcal{A}^n \rightarrow \mathcal{Q} $ where $\mathcal{Q} = \{B \in [0,1]^{n \times n}: B 1 = 1 \}$ is the set of right stochastic matrices, and such that $Q(A)_{ij} > 0$ only if $A_{ij} > 0$. The function $Q(.)$ maps an adjacency matrix to a transition matrix of Markov chain. We will let
\begin{equation}
\CPR{p(t+1) = j}{p(t) = i, A(t) = A_t} = Q(A_t)_{ij}.
\end{equation}
With this construction we have $A(t)_{p(t),p(t+1)} = 1$, so that the token only moves in available connections.

We now present some results that will be useful later when we analyze the convergence property of our algorithm. If assumption \ref{asp:networkA} holds then $A(t)$ is an \textit{i.i.d.} sequence, and we can compute\cmt{justify that $A(t)$ and $p(t)$ are independent}
\begin{align*}
 \MoveEqLeft \CPR{p(t+1) = j}{p(t) = i}  \\
& = \sum_{A_k} \CPR{p(t+1) = j}{p(t) = i, A(t) = A_k} \\  
& \qquad\qquad\qquad \CPR{A(t) = A_k}{p(t) = i}  \\ 
& = \sum_{A_k} Q(A_k)_{ij} \PR{A(t) = A_k} = \EX{Q(A(t))}_{ij},
\end{align*}
so that $p(t)$ is a time invariant Markov chain, with transition matrix $\overline{Q} = \EX{Q(A(t))}$. 

We wish to study conditions under which the transition matrix $\overline{Q}$ is irreducible. Consider the following
\begin{assumption}
\label{asp:matrixQ}
The function $Q(.)$ is such that if $A_{ij} > 0$ then $Q(A)_{ij} \geq \delta > 0$. 
\end{assumption}
\noindent
An example of a function $Q(.)$ satisfying assumption \ref{asp:matrixQ} is the one which assigns to each non-zero entry in row $i$ the value ${\degree(i,A)}^{-1}$, where $\degree(i,A) = e_i^\top A \one$ is the outdegree of node $i$ according to $A$.

\begin{lemma}($\overline{Q}$ is irreducible)
\label{lem:irreducible}
Suppose assumptions \ref{asp:networkA} and \ref{asp:matrixQ} hold. Then the matrix $\overline{Q}$ is irreducible. 
\end{lemma}
\begin{proof}
The matrix $\overline{Q}$ is irreducible if and only if the graph $\mathcal{G}(\overline{Q})$ is strongly connected (for a justification of this see \cite[p.~671]{meyer2000matrix}). Since by assumption \ref{asp:matrixQ}, $Q(A)_{ij} > 0$ if $A_{ij} > 0$, the graph $\mathcal{G}(\overline{Q})$ will contain all the edges of $\mathcal{G}(\EX{A(t)})$. We can see that this is the case by writing $\EX{A(t)} = \sum_{A_k \in \mathcal{A}^n}A_k \PR{A(t) = A_k}$ and thus if $\EX{A(t)}_{ij} > 0$ then there is a transition matrix $B$ with $B_{ij} > 0$ and $\PR{A(t) = B} > 0$. Now, we have that $\EX{Q(A(t))} \geq Q(B) \PR{A(t) = B} > 0$. Since by assumption \ref{asp:networkA}, $\mathcal{G}(\EX{A(t)})$ is strongly connected, $\mathcal{G}(\overline{Q})$ is also strongly connected.
\end{proof}

\section{Token estimator}
\label{sec:particle}

In this section we present our algorithm for the estimation of $\theta$. We consider $(p(t))_{t \geq 0}$ and $(w(t))_{t \geq 0}$ as defined in the previous sections, and define the natural filtrations $\mathcal{M}_t = \sigma(p(0) \ldots, p(t))$ and $\mathcal{F}_t = \sigma(\mathcal{M}_0,w(0),\ldots,\mathcal{M}_t,w(t))$. We now define the processes:
\begin{itemize}
\item $\tau_i(t)  = \sup \left\{k \leq t: p(k) = i \right\}$, the last visitation time ($\leq t$) of the token to state $i$
\item $S(t) = \left\{i \in \mathcal{V}: \exists _{0 \leq k \leq t}~ p(k) = i \right\}$, the set of visited nodes at time $t$
\item $\IND{i \in S(t)}$, indicator random variable is $1$ if and only if node $i$ has already been visited at time $t$ 
\end{itemize}
and note that they are all adapted to $\mathcal{M}_t$. 

In our algorithm the traveling token carries with it two quantities, the vector of recent local variables $d(t)$, and a matrix $K(t)$ that is an approximation to the Fisher information matrix $\Sigma_c$. Specifically, in the algorithm we propose, the agent with the token can compute at time instant $t$ an estimate $\test(t)$ given by
\begin{equation}
\label{eq:estimatedef}
\test(t) = (I\alpha(t)^{-1} + K(t))^{-1} d(t)
\end{equation}
where
\begin{align*}
K(t) & = \sum_{i=1}^n H_i^\top C_i^{-1} H_i \IND{i \in S(t)} \\
	&= \sum_{i=1}^n B_i \IND{i \in S(t)}
\end{align*}
where we defined $B_i = H_i^\top C_i^{-1} H_i$. The deterministic sequence $\alpha(t)$ is chosen so that $\alpha(t) > 0,$ for all $t$, and thus the term inside the parenthesis is invertible as long as assumption \ref{asp:invert} holds. The vector $d(t)$ is defined in a similar way as before, so that
\begin{align*}
d(t)  & = \sum_{i=1}^n H_i^\top C_i^{-1} \overline{y}_i(\tau_i(t)) \IND{i \in S(t)} \\
  & = \sum_{i=1}^n x_i(\tau_i(t)) \IND{i \in S(t)}.
\end{align*}
We note the inclusion of the indicator variable $\IND{i \in S(t)}$ to capture the fact that before a node has been visited, the associated quantities in $d(t)$ and $K(t)$ are $0$. 

\paragraph*{Implementation} The algorithm $1$ presented below represents a possible implementation of our estimator in a distributed network. At the beginning of each time instant $t$, agents take a new measurement of the environment according to \eqref{eq:measurement}, and update the local variable $x_i(t)$ with that new measurement as 
\begin{equation}
x_i(t) = x_i(t-1) + \frac{1}{t}\left(H_i C_i^{-1} y_i(t) - x_i(t-1) \right) 
\end{equation}
so that at time $t$ the variable $x_i(t)$ equals the quantity $H_i^\top C_i^{-1} \overline{y}_i(t)$. This is the only action that the agents perform if they do not hold the token. After this, the node that contains the token, node $i = p(t)$, will update the variables $\{K(t),d(t)\}$ with the agent information. The variable $K(t)$ is updated only if the agent has not been visited yet, as $K(t) = K(t-1) + B_i$, and then $d(t)$ is updated as $d(t) = d(t-1) - \tilde{x}_i(t) + x_i(t)$, where the variable $\tilde{x}_i(t)$ is used to save locally the value of $x_i(t)$ in the last time the particle visited node $i$. Finally the node with the token updates the local variable $\tilde{x}_i(t) = x_i(t)$. The estimate can be obtained at node $p(t)$ as $\test(t) = \left( I\alpha(t)^{-1} + K(t) \right)^{-1} d(t)$. At the end of time instant $t$, the quantities $(K(t),d(t))$ are passed to a neighbour of $i$ at random according to the probabilities in $Q(A(t))$, so that node $j$ will be chosen as node $p(t+1)$ with probability $Q(A(t))_{ij}$.
We note that the function $Q(.)$ can be taken to depend only on the local information available to the agent, for example by selecting $Q(A(t))_{ij} = \degree(A(t),i)^{-1}$ for all $j$ such that $A(t)_{ij} > 0$. 

\begin{algorithm}
\caption{Algorithm \algname}\label{alg:particle_algorithm}
\begin{algorithmic}[1]
\State All agents $i \in \mathcal{V}$: set  $x_i = 0,\tilde{x}_i = 0$
\State User: set $p(0)$ (starting node of token)
\State Agent $p(0)$: Set $d = 0; K = 0$
\Loop ~ at each time instant $t$ (each agent $i$ independently)
\State Sample $y_i$
\State $x_i \gets x_i + \frac{1}{t}\left(H_i^\top C_i^{-1} y_i - x_i \right) $
\If{$i = p(t)$}
\State $d \gets d - \tilde{x}_i + x_i$
\State $\tilde{x}_i \gets x_i$
\If{this is the first visit to node $i$}
\State $K \gets K + B_i$
\EndIf
\State $\hat{\theta} = \left(I\alpha(t)^{-1} + K \right)^{-1} d$
\State Sample $j$ according to $Q(A(t))$
\State node $i$ sends $\{K,d\}$ to node $j$
\State $p(t+1) = j$
\EndIf
\EndLoop
\end{algorithmic}
\end{algorithm}

\section{Main Results}
\label{sec:proof}

In this section we present the main theoretical results concerning the algorithm \algname. We will prove two theorems: theorem $\ref{thm:consistentA}$ states that our estimator is consistent, and theorem \ref{thm:theorem} states that our estimator achieves the central estimator variance asymptotically. 

\begin{lemma}(Condition for estimator being consistent) 
\label{lem:consistencyinfinitelly}
Suppose the Markov chain is such that with probability $1$ all nodes are visited infinitely often. Then if $\alpha(t)$ is chosen so that $\alpha(t)^{-1} \rightarrow 0$, estimator $\test(t)$ as defined in equation \eqref{eq:estimatedef} is consistent, $\lim_{t \rightarrow + \infty} \test(t) = \theta \as$.
\end{lemma}
\begin{proof}
By the law of large numbers $\overline{y}_i(t) \rightarrow H_i \theta \as$. Because the nodes are visited infinitely often with probability $1$, we have that $\tau_i(t) \rightarrow +\infty \as$, $\IND{i \in S(t)} \rightarrow 1 \as$. Thus, $x_i(\tau_i(t))\IND{i \in S(t)} \rightarrow H_i^\top C_i^{-1} H_i \theta \as$ and $$\left(I\alpha(t)^{-1} + \sum_{i=1}^n B_i\IND{i \in S(t)}\right) \rightarrow \sum_{i=1}^n H_i^\top C_i^{-1} H_i \as$$  
\end{proof}

\begin{theorem}(Estimator is consistent under assumption \ref{asp:networkA})
\label{thm:consistentA}
Let assumptions \ref{asp:noise},\ref{asp:invert} hold, and suppose the network instantiations are such that assumption \ref{asp:networkA} holds. Suppose the function $Q(.)$ is such that assumption \ref{asp:matrixQ} holds. Then the estimate $\test(t)$ is consistent.  
\end{theorem}
\begin{proof}
In view of lemma \ref{lem:irreducible} the Markov chain is irreducible, and an irreducible Markov chain visits all nodes infinitely often with probability $1$. In view of lemma \ref{lem:consistencyinfinitelly} the estimate $\test(t)$ is consistent. 
\end{proof}

Now we study the asymptotic mean square error of our estimator. For that purpose we will make use of the following lemmas, which are used in our proof to establish asymptotic properties of the random variables $\tau_i(t)$. 

\begin{lemma}(Exponential tail for stopping times)
\label{lem:exponentialtail}
Let $T$ denote a stopping time \textit{w.r.t.} a filtration $\mathcal{F}_t$. Suppose that for some $m \in \mathbb{N}, \epsilon > 0$, it holds that $\CPR{T \leq t + m}{\mathcal{F}_t} \geq \epsilon$, for all $t \geq t_0$. Then we have $\PR{T > t} \leq (1 - \epsilon)^{\frac{t-t_0}{m} - 1}$
\end{lemma}
\begin{proof}
This is a standard result for stopping times. See \textit{e.g.} \cite[chapter~E10.5]{williams1991probability}.
%\begin{align*}
%& \PR{T > t_0 + km}  = \PR{T > t_0 + km, T > t_0 + (k-1)m} \\ 
%& = \EX{\CPR{T > t_0 + km, T > t_0 + (k-1)m}{\mathcal{F}_{t_0 + (k-1)m}}} \\
%& = \PR{T > t_0 + (k-1)m} \EX{\CPR{T > t_0 + km}{\mathcal{F}_{t_0 + (k-1)m}}} \\
%& \leq  \PR{T > t_0 + (k-1)m} (1-\epsilon).
%\end{align*}
%By induction, $\PR{T > t_0 + km} \leq (1- \epsilon)^k$. Then we have $\PR{T > t} \leq \PR{T > t_0 %+ \lfloor \frac{t - t_0}{m} \rfloor} \leq (1 - \epsilon)^{\frac{t- t_0}{m} - 1}. $
\end{proof}

For any set $E \subset \mathcal{V}$, and any $t_0$, we would like to upper bound the probability $\PR{p(t_0) \not\in E, \ldots, p(t_0 + t) \not \in E}$. We can define $T_E^{t_0} = \inf\{t \geq t_0: p(t) \in E \}$, and then $T_E^{t_0}$ is a stopping time. Further, $\PR{p(t_0) \not\in E, \ldots, p(t_0 + t) \not \in E} = \PR{T_E^{t_0} > t}$. Thus all we need is to show that for some $m,\epsilon$ we have $\CPR{T_E^{t_0} \leq t + m}{\mathcal{F}_t} \geq \epsilon$ for all $t \geq t_0$ and apply lemma \ref{lem:exponentialtail}. 

\begin{lemma}
\label{lem:expontentialtaillemma}
Suppose the network instantiations are such that assumption \ref{asp:networkA} holds, and suppose the function $Q(.)$ is such that assumption \ref{asp:matrixQ} holds, so that $Q(A)_{ij} \geq \delta$ if $A_{ij} > 0$. Then, by taking $m = n$, $\epsilon = \delta^n$ we have $\CPR{T_E^{t_0} \leq t + m}{\mathcal{F}_t} \geq \epsilon$ for any non-empty set $E$, and all $t \geq t_0$.  
\end{lemma}
\begin{proof}
\begin{align*}
& \CPR{T_E^{t_0} \leq t + n}{\mathcal{F}_t}  \\
& = \CPR{\{p(t_0) \in E \} \cup \ldots \cup \{p(t + n) \in E\}}{\mathcal{F}_t}  \\
& \geq \CPR{\{p(t+1) \in E \} \cup \ldots \cup \{p(t+n) \in E\}}{\mathcal{F}_t} \\
&= \sum_{i \in \mathcal{V}} \CPR{\{p(t+1) \in E \} \cup \ldots \cup \{p(t+n) \in E\}}{p(t) = i} \\ 
& \qquad\qquad\qquad\qquad\qquad\qquad\qquad \CPR{p(t) = i}{\mathcal{F}_t} \\
& \geq \delta^n
\end{align*}
where to establish the last inequality we note that since the Markov chain is irreducible, there is a path connecting any two nodes $i,j$ of length at most $n$. By choosing a $j \in E$, we can obtain a path of length at most $n$ which is guaranteed to visit the set $E$. Given assumption \ref{asp:matrixQ} the transition probabilities in this path are all lower bounded by $\delta$ and this yields the inequality.
\end{proof}

\begin{lemma}
\label{lem:tailforvariables}
Suppose the network instantiations are such that assumption \ref{asp:networkA} holds and suppose the function $Q(.)$ is such that assumption \ref{asp:matrixQ} holds. Then for any node $i \in \mathcal{V}$ we have
\begin{itemize}
\item $\PR{i \not \in S(t)} \leq c_1 \E^{-c_2 t}$, for all $t \geq 0$
\item $\PR{S(t) \neq \mathcal{V}} \leq n c_1 \E^{-c_2 t}$, for all $t \geq 0$
\item $\PR{p(t_0) \neq i , \ldots, p(t_0 + t) \neq i} \leq c_1 \E^{-c_2 (t - t_0)}$, for all $t \geq t_0 \geq 0$. 
\end{itemize}
where the constants $c_1,c_2 > 0$ do not depend on $i,t,t_0$. 
\begin{proof}
In view of lemmas \ref{lem:expontentialtaillemma} and \ref{lem:exponentialtail} we know that for some $m,\epsilon$, we have $\PR{T_E^{t_0} > t} \leq (1 - \epsilon)^{\frac{t-t_0}{m} - 1}$ for all $t_0$, $t \geq t_0$. We can set $E = \{i\}$. Since $\PR{i \not \in S(t)} = \PR{T_{\{i\}}^0 > t}$ it follows, $\PR{i \not \in S(t)} \leq (1 - \epsilon)^{\frac{t}{m} - 1}$ and we can select $c_1 = (1 - \epsilon)^{-1}$ and $c_2 = \frac{\log(1 - \epsilon)}{m}$. Then we also have $\PR{S(t) \neq \mathcal{V}} \leq \sum_{i=1}^n \PR{i \not \in S(t)} \leq n c_1 \E^{-c_2 t}$. Finally, since $\PR{p(t_0) \neq i, \ldots, p(t_0 + t) \neq i} = \PR{T_{\{i\}}^{t_0} > t}$ we have $\PR{p(t_0) \neq i, \ldots, p(t_0 + t) \neq i} \leq (1 - \epsilon)^{\frac{t-t_0}{m} - 1}$, and selecting the same constants $c_1$ and $c_2$ work.
\end{proof}

\end{lemma}

\begin{theorem}(Estimator is asymptotically optimal under assumption \ref{asp:networkA})
\label{thm:theorem}
Let assumptions \ref{asp:noise},\ref{asp:invert} hold, and suppose the network instantiations are such that assumption \ref{asp:networkA} holds. Suppose the function $Q(.)$ is such that assumption \ref{asp:matrixQ} holds. Suppose the sequence $\alpha(t)$ is chosen so that 
\begin{align*}
\lim_{t \rightarrow +\infty} \frac{t \alpha(t)^2}{\E^{c_2 t}} = 0, \;\; \lim_{t \rightarrow +\infty} \frac{t}{\alpha(t)^2} = 0, 
\end{align*}
where $c_2$ is as defined in lemma \ref{lem:tailforvariables}. An example of such a choice is $\alpha(t) = t$, which verifies both conditions for any $c_2 > 0$. Then estimator $s(t)$ is asymptotically optimal, that is, we have $$\lim_{t \rightarrow +\infty} t \EX{(\test(t) - \theta)(\test(t) - \theta)^\top} = \Sigma_c^{-1}.$$
\end{theorem}
\begin{proof}
We define $\widetilde{K}(t) = \left( I \alpha(t)^{-1} + \sum_{i=1}^N A_n \IND{n \in S_t} \right) $ and write
\begin{align*}
& s(t) - \theta \\
& = \widetilde{K}(t)^{-1} \left( \sum_{i=1}^n H_i^\top C_i^{-1} \overline{y}_i(\tau_i(t))\IND{i \in S(t)}\right) - \theta  \\
& = \widetilde{K}(t)^{-1} \left( \sum_{i=1}^n H_i^\top C_i^{-1} \overline{y}_i(\tau_i(t))\IND{i \in S(t)} - H_i^\top C_i^{-1} H_i \theta \right) \\
& \qquad\qquad\qquad + \left( \widetilde{K}(t)^{-1} \sum_{i=1}^n H_i^\top C_i^{-1} H_i - I \right) \theta  \\ 
& = \widetilde{K}(t)^{-1} \left( \sum_{i=1}^n x_i(\tau_i(t))\IND{i \in S(t)} - B_i \theta \right) \\
& \qquad\qquad\qquad + \left( \widetilde{K}(t)^{-1}\Sigma_c - I \right) \theta.
\end{align*}
Then we can write
\begin{align}
& t(s(t) - \theta)(s(t) - \theta)^\top  \nonumber\\
& = t \widetilde{K}(t)^{-1} \left( \sum_{i=1}^n x_i(\tau_i(t))\IND{i \in S(t)} - B_i \theta \right) \nonumber \\
& \qquad\qquad \left( \sum_{i=1}^n x_i(\tau_i(t))\IND{i \in S(t)} - B_i \theta \right)^\top \widetilde{K}(t)^{-1}  \label{eq:firstterm}\\
& + t \widetilde{K}(t)^{-1} \left( \sum_{i=1}^n x_i(\tau_i(t))\IND{i \in S(t)} - B_i \theta \right)  \theta^\top \nonumber\\ 
& \qquad\qquad\qquad\qquad \left( \widetilde{K}(t)^{-1}\Sigma_c - I \right)^\top   \label{eq:secondterm}\\
& + t \left( \widetilde{K}(t)^{-1}\Sigma_c - I \right) \theta  \left( \sum_{i=1}^n x_i(\tau_i(t))\IND{i \in S(t)} - B_i \theta \right)^\top \nonumber\\
& \qquad\qquad\qquad\qquad\qquad\qquad \widetilde{K}(t)^{-1}  \nonumber\\
& + t \left( \widetilde{K}(t)^{-1}\Sigma_c - I \right) \theta \theta^\top \left( \widetilde{K}(t)^{-1}\Sigma_c - I \right)^\top  \label{eq:fourthterm} \\ 
& = U(t) + R_1(t) + R_1(t)^\top + R_2(t) \nonumber
\end{align}
where we introduced  the variables $U(t)$, $R_1(t)$ and $R_2(t)$ corresponding to the terms \eqref{eq:firstterm}, \eqref{eq:secondterm} and \eqref{eq:fourthterm} respectively to simplify notation. We will work with each term separately, starting with $U(t)$ and showing that its expected value converges to $\Sigma_c^{-1}$. Then in appendix \ref{ap:residuals} we show that the remaining terms $R_1(t)$ and $R_2(t)$ go in expectation to $0$. 

We can begin by computing the expected value of $U(t)$ given the filtration $\mathcal{M}_t$, $\CEX{U(t)}{\mathcal{M}_t}$. In order to do that, we evaluate for  $i,j \in \mathcal{V}$ the expression
\begin{align}
&\mathbb{E}\biggl[ \left( x_i(\tau_i(t))\IND{i \in S(t)} - B_i \theta \right) \nonumber \\ \label{eq:expression}
& \qquad\qquad\qquad \left( x_j(\tau_j(t))\IND{j \in S(t)} - B_j \theta \right)^\top \biggr\rvert \mathcal{M}_t \biggr]
\end{align}
We note, $\tau_i(t)$,$\tau_j(t)$ and $\IND{i \in S(t)},\IND{j \in S(t)}$ are all $\mathcal{M}_t-$measurable, and $(x_i(t))_{t \geq 0}, (x_j(t))_{t \geq 0}$ are both independent of $\mathcal{M}_t$. Since for any $i \in \mathcal{V}$, 
$
\EX{x_i(t) - B_i \theta} = 0$
and
$$
\EX{\left( x_i(t) - B_i \theta \right) \left( x_i(t) - B_i \theta \right)^\top} = \frac{H_i^\top C_i^{-1} H_i}{t}
$$
the expression \eqref{eq:expression} evaluates as $B_i \theta \theta^\top B_j $ if $\left\{ i \not \in S(t) \wedge j \not \in S(t) \right\}$, $0$ if $i \neq j $ and $\{i \in S(t) \vee j \in S(t)\}$ since $x_i(t)$ and $x_j(t)$ are independent, and as  
$
\frac{H_i^\top C_i^{-1} H_i}{\tau_i(t)}
$
if $i = j$ and $\left\{i \in S(t)\right\}$. Thus, we have
\begin{align}
\label{eq:cexfixedpath}
\MoveEqLeft \CEX{U(t)}{\mathcal{M}_t}  \nonumber \\
& = t \widetilde{K}(t)^{-1} \sum_{i=1}^n \frac{H_i^\top C_i^{-1} H_i}{\tau_i(t)} \IND{n \in S(t)} \widetilde{K}(t)^{-1}  \\
&  + \widetilde{K}(t)^{-1} t \sum_{i,j \in \mathcal{V}} B_i \theta \theta^\top B_j \IND{i,j \not \in S(t)} \widetilde{K}(t)^{-1} \label{eq:cexfix2} \\ 
& = U_1(t) + U_2(t) \nonumber
\end{align}
where we associate $U_1(t)$, $U_2(t)$ with the expressions in \eqref{eq:cexfixedpath}, \eqref{eq:cexfix2} respectively to simplify notation.

First we show that the expectation of $U_2(t)$ converges to zero. We can establish that
\begin{align*}
\norm{ \widetilde{K}(t)^{-1}} & = \norm{ \left(I\alpha(t)^{-1} + \sum_{i=1}^n B_i \IND{i \in S(t)}\right)^{-1}} \\
& = \rho\left(\left(I\alpha(t)^{-1} + \sum_{i=1}^n B_i \IND{i \in S(t)}\right)^{-1}  \right) \\
&\leq \alpha(t)
\end{align*}
where we note that the norm equals the spectral radius since the matrix is symmetric, and in establishing the last bound we used the property that for any real $k$ we have
\begin{equation}
\label{eq:sumI}
\lambda \in \eig(A) \iff (k + \lambda) \in \eig(k I + A),
\end{equation}
and noting that the matrices $B_i$ are all positive semidefinite. Then we can write
\begin{align*}
 \norm{U_2(t)} & \leq t \norm{\widetilde{K}(t)^{-1}}^2 \norm{\sum_{i,j \in \mathcal{V}} B_i \theta \theta^\top B_j} \IND{S(t) \neq \mathcal{V}}  \\
& \leq t \alpha(t)^2 \norm{\sum_{i,j \in \mathcal{V}} B_i \theta \theta^\top B_j} \IND{S_t \neq \mathcal{V}}.
\end{align*} 
Then we have
\begin{align}
\label{eq:normedterm}
\MoveEqLeft \EX{\norm{\CEX{U_2(t)}{\mathcal{M}_t}}}  \nonumber \\
&  \qquad \leq t \alpha(t)^2 \norm{\sum_{i,j \in \mathcal{V}} B_i \theta \theta^\top B_j} \PR{S(t) \neq \mathcal{V}}.
\end{align}
In view of lemma \ref{lem:tailforvariables}, $\PR{S(t) \neq \mathcal{V}} \leq n c_1 \E^{-c_2t}$, and given our assumptions on the chosen sequence $\alpha(t)$, we can take the limit on \eqref{eq:normedterm} to conclude $\lim_{t \rightarrow +\infty}\EX{\norm{\CEX{U_2(t)}{\mathcal{M}_t}}} = 0$, from which follows that $\lim_{t \rightarrow +\infty} \EX{U_2(t)} = 0$.

Now we  will show that $\lim_{t \rightarrow +\infty}\EX{U_1(t)} = \Sigma_c^{-1}$. We write
\begin{align*}
\EX{U_1(t)} & = \CEX{U_1(t)}{S(t) = \mathcal{V}}\PR{S(t) = \mathcal{V}}  \\ 
& \qquad + \CEX{U_1(t)}{S(t) \neq \mathcal{V}}\PR{S(t) \neq \mathcal{V}}  
\end{align*}
and show that the second term has limit $0$. We can do this in a similar way as for the term $U_2$, by upper bounding the norm as 
$
\norm{\CEX{U_1(t)}{S(t) \neq \mathcal{V}}\PR{S(t) \neq \mathcal{V}}} \leq t \norm{\widetilde{K}(t)^{-1}}^2 \norm{\Sigma_c}\PR{S(t) \neq \mathcal{V}}
$
and given the results from lemma \ref{lem:tailforvariables} this establishes an upper bound that decays to $0$. Finally, consider the term 
\begin{align*}
&\CEX{U_1(t)}{S(t) = \mathcal{V}} = \left(I\alpha(t)^{-1} + \Sigma_c\right)^{-1} \\ & \qquad\qquad \sum_{i=1}^n B_i \CEX{\frac{t}{\tau_i(t)}}{S(t) = \mathcal{V}} \left(I\alpha(t)^{-1} + \Sigma_c\right)^{-1}. 
\end{align*}
We will show that $\lim_{t\rightarrow +\infty}\CEX{\frac{t}{\tau_i(t)}}{S(t) = \mathcal{V}} = 1$, from which will follow that
\begin{align*}
\lim_{t\rightarrow +\infty} \EX{U_1(t)} = \Sigma_c^{-1} \left( \sum_{i=1}^n B_i \right) \Sigma_c^{-1} = \Sigma_c^{-1}. 
\end{align*}

Since if $S(t) = \mathcal{V}$ we have that $\frac{t}{\tau_i(t)} \geq 1$, we just need to show that $$\limsup_{t\rightarrow +\infty}\CEX{\frac{t}{\tau_i(t)}}{S(t) = \mathcal{V}} \leq 1. $$ We can write for any positive $b < t$
\begin{align}
& \CEX{\frac{t}{\tau_i(t)}}{S(t) = \mathcal{V}}   \nonumber \\
& = \CEX{\frac{t}{\tau_i(t)} }{\tau_i(t) \geq t - b,~S(t) = \mathcal{V}} \PR{\tau_i(t) \geq t - b}  \nonumber\\
&  \qquad + \CEX{\frac{t}{\tau_i(t)}}{\tau_i(t) < t - b,~S(t) = \mathcal{V}} \PR{\tau_i(t) < t - b}  \nonumber\\
& \leq \frac{t}{t-b + 1} + t \PR{\tau_i(t) < t - b}. \label{eq:nonvisitbound}
\end{align}
We have that 
%\PR{\tau_n(t) \leq t - b} & = \PR{p_{t-b+1} \neq n \wedge p_{t-b+2} \neq n \wedge \ldots \wedge p_{t} \neq n} \nonumber\\
% & = \PR{p_{1} \neq n \wedge p_{2} \neq n \wedge \ldots \wedge p_{b} \neq n} \label{eq:stationaryassump}\\
% & = \PR{n \not\in S_b} \nonumber
%\end{align}
\begin{align*}
\MoveEqLeft \PR{\tau_i(t) < t - b}  \\
& \; = \PR{p_{t-b} \neq i \wedge p_{t-b+1} \neq i \wedge \ldots \wedge p_{t} \neq i}  \\
& \;  \leq  c_1 \E^{-c_2 b}
\end{align*}
for some positive constants $c_1,c_2 > 0$, where we used in the last inequality the result from lemma \ref{lem:tailforvariables}. We can then bound \eqref{eq:nonvisitbound} as
\begin{align}
\label{eq:limittobetaken}
\EX{\frac{t}{\tau_i(t)} \IND{i \in S(t)}} & \leq \frac{t}{t-b+1} + t c_1  \E^{-c_2 b} \quad.
\end{align}
Since this is valid for any $t> b > 0$, we can choose for each $t$, $b(t) = \sqrt{t}$. Then by taking the limit on \eqref{eq:limittobetaken} we have 
\begin{align*}
\limsup_{t \rightarrow +\infty} \EX{\frac{t}{\tau_i(t)} \IND{i \in S(t)}} \leq 1 .
\end{align*}
\end{proof}

\section{Simulation}
\label{sec:simulation} 

In this section we implement our algorithm and test its performance by means of a numerical simulation. We will focus exclusively on random, $\textit{i.i.d.}$ network instantiations, and perform comparisons with the \textit{consensus+innovations} algorithm \cite{kar2011convergence, kar2013consensus+, kar2013distributed}.

\paragraph*{Choice of Markov chain}

In order to implement our algorithm we need to choose the weights of the Markov chain, which correspond to the choice of function $Q(.)$. According to the results of theorem \ref{thm:theorem} any choice satisfying assumption \ref{asp:matrixQ} will guarantee asymptotic convergence to the central estimator variance. Naturally, for finite $t$, the convergence speed for a particular network setting will depend on the choice of $Q(.)$. It is outside of the scope of this paper to examine the performance of our algorithm in a finite time window as a dependency of the choice of $Q(.)$, and instead we focus on showing empirically that for some choice of transition matrix the algorithm performs well. In our tests we always use a Markov chain with transition probabilities equal to the reciprocal of the out-degree of the agent, so that for each $i$, $Q(A)_{ij} = \degree(A,i)^{-1} = (e_i A \one)^{-1}$, for all $j$ such that $A_{ij} > 0$. This choice of weights is appropriate for a distributed setting, as all an agent needs to know is the number of neighbours it has in order to compute the transition probabilities. 

\paragraph*{Consensus+innovations}

We will be comparing our algorithm with a \textit{consensus+innovations} type algorithm. Specifically, we will consider that its iterations are given by 
\begin{align*}
s_i(t+1) & = s_i(t) - \beta(t) \left(\sum_{l \in \Omega_i(t)} s_i(t) - s_l(t) \right) \\ & + \alpha(t) K_i(t) H_i^\top C_i^{-1} \left( y_i(t) - H_i s_i(t) \right),
\end{align*}
where $\alpha(t) = \frac{a}{(t+1)^{\tau_1}}$ and $\beta(t) = \frac{b}{(t+1)^{\tau_2}}$, $\Omega_i(t)$ denotes the neighbours of $i$ at time $t$. The gain $K_i(t)$ is computed by means of another iterative procedure \cite{kar2013consensus+}. For an appropriate choice of $\tau_1$,$\tau_2$ and $a$, $s_i(t)$ is guaranteed to have an asymptotic mean square error (MSE) equal to the central estimator; however, we don't have a way of determining optimal values for these constants. In our tests we did a grid search to find for each setting, the set of parameters with best performance, while guaranteeing that they stay in the range that guarantees asymptotic optimality.

In the \textit{consensus+innovations} algorithm, all agents have a local estimate that is constantly updated, whereas in our algorithm only one estimate exists in the network, localized at the node that is currently carrying the token. In a practical scenario, it may be required that all agents have a online estimate during the running time of the algorithm, at all times. In this case, we can extend our algorithm by having each agent save the estimate from the last time it was visited by the token. Using our notation, this would mean that each agent $i$ has a saved estimate equal to $s(\tau_i(t))$ at time $t$. 

In our tests we will be comparing the mean square error (MSE) between the different algorithms. We look at the three quantities\cmt{needs to be edited in view of comment: 'write these as three equalities, each with the proper acronym in the left-hand side' - is it worthwhile defining three acronyms that we only use once?} 
\begin{align*}
\frac{1}{n}\frac{\sum_{i=1}^n \EX{\norm{s_i(t) - \theta}^2}}{\norm{s_i(0) - \theta}^2}, \;\; \frac{ \EX{\norm{s(t) - \theta}^2}}{\norm{s(0) - \theta}^2} 
\end{align*}
and
\begin{align*}
\EX{\frac{1}{|S(t)|} \frac{\sum_{i=1}^n \norm{s(\tau_i(t)) - \theta}^2}{\norm{s(0) - \theta}^2}},
\end{align*}
respectively, the relative mean square error (r-MSE) of a network running a \textit{consensus+innovations} algorithm, the r-MSE of the token estimate, and the r-MSE of a network that runs the token algorithm where each agent saves the last estimate seen.  

\paragraph*{Experimental results}

In order to test our algorithm, we randomly generate two geometric graphs of size $n = 20$ and $n = 50$, with relative degrees $0.12$ and $0.09$ respectively. In each case, we consider that they represent the backbone of our network, and at any time instant, a communication link can fail with a fixed probability $0.5$. We will test each network under two different measurement models for agents. Under model $A$, we let $L = \frac{n}{4}$ (rounded to the closest integer) and have $ \theta = [1 \ldots L]^\top $, and $H_i \in \mathbb{R}^{1 \times L}$, with entries randomly generated from a normal distribution. Under model $B$, $\theta = [1\; 2\; \ldots\; n ]^\top$, and $H_i \in \mathbb{R}^{1 \times n}$, with entries taken randomly from a normal distribution. In both cases, the noise tested is simple Gaussian noise with variance $1$, independent between the agents, $w_i(t) \sim \mathcal{N}(0,1)$. 
\begin{figure}[!t]
\centering
\includegraphics[trim={1cm 1.1cm 0.9cm 0.8cm},clip,width=3in]{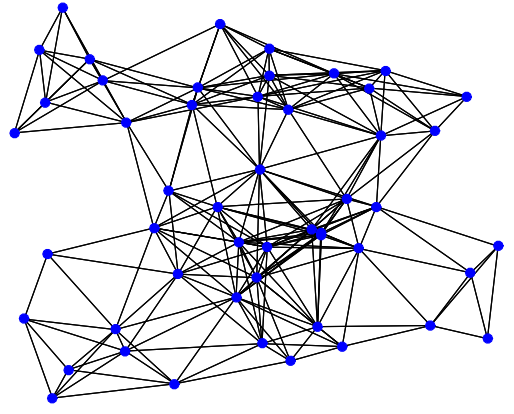}
\caption{Geometric network with $n=50$ used as the backbone network in our tests.}
\label{fig:geonetwork}
\end{figure}
For illustration purposes, we present in \figurename \ref{fig:geonetwork} the $50$ node backbone network. 

In \figurename \ref{fig:20nodesnonobservable} and \figurename \ref{fig:50nodes_nonobervablerandom} we present the mean-square error at each iteration time, for the $20$ node network and the $50$ node network respectively, considering in each case both measurement models A and B. The relative mean square error was obtained by averaging the square error over multiple runs.  We note that our algorithm showed in all tested cases  significant improvement with respect to the \textit{consensus+innovations}, both when comparing the token carried estimate to the consensus+innovations network average, and when comparing the network average. Further, we note our algorithm has the advantage of not requiring the tuning of any parameters, and uses less communication resources per iteration than \textit{consensus+innovations}. 

\begin{figure}[!t]
\centering
\includegraphics[trim={0.4cm 0.1cm 0cm 0.2cm},clip,width=3.5in]{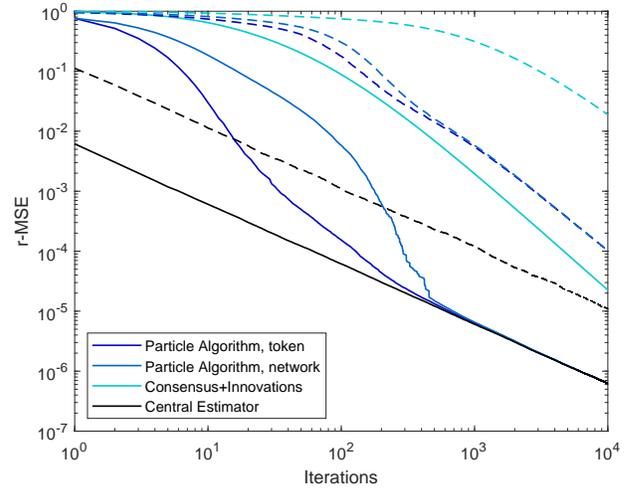}
\caption{Relative mean square error for the $20$ node network. Continuous line correspond to model A, dashed line correspond to model B.}
\label{fig:20nodesnonobservable}
\end{figure}
\begin{figure}[!t]
\centering
\includegraphics[trim={0.4cm 0.1cm 0cm 0.2cm},clip,width=3.5in]{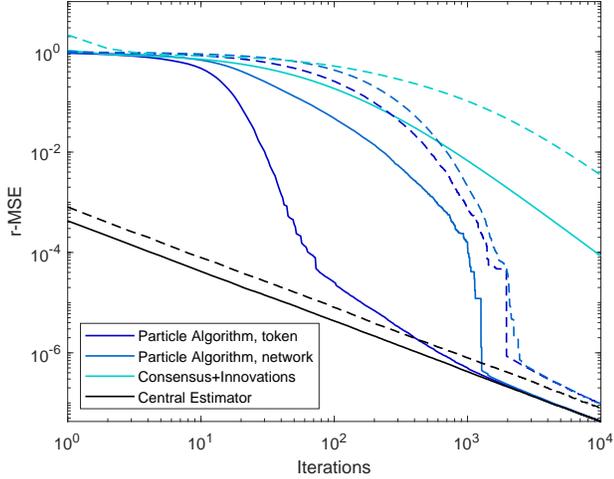}
\caption{Relative mean square error for the $50$ node network. Continuous line correspond to model A, dashed line correspond to model B.}
\label{fig:50nodes_nonobervablerandom}
\end{figure}

\section{Other assumptions on the network instantiations}
\label{sec:networkB}

In this section we explain how the proof for consistency and asymptotic optimality of our estimator can be extended to other network connection settings, and in particular we consider the following.
\begin{assumption}
\label{asp:networkB}
The sequence $(A(t))_{t \geq 0}$ is deterministic, with the property that for some positive integer $b$, the graph $\bigcup_{t \leq k < t + b} \mathcal{G}(A(k))$ is strongly connected, for all $t$. 
\end{assumption}
\noindent This assumption appears in a wide range of other works in distributed algorithms, and in particular in \cite{ram2009incremental,ram2010distributed}. 

According to lemma \ref{lem:consistencyinfinitelly}, in order to show consistency all we need is to show that for an appropriate choice of $Q(.)$, all nodes are visited infinitely often. This is explained below in theorem \ref{thm:theoremconsistencyB}. 

In order to show asymptotic optimality, the only part of the proof of theorem \ref{thm:theorem} that depended on the network instantiations was captured in the bounds established in lemma \ref{lem:tailforvariables}. These bounds depend exclusively on the result established on lemma \ref{lem:expontentialtaillemma}, namely, that given our network instantiations, and the choice of function $Q(.)$, we could guarantee that for some $m,\epsilon$, the inequality 
\begin{equation}
\label{eq:boundonT}
\CPR{T^{t_0}_E \leq t + m}{\mathcal{F}_t} \geq \epsilon
\end{equation}
hold, for all $t \geq t_0$. Hence, it is sufficient to show that this bound holds to show asymptotic optimality of our estimator. 

%In fact, it is sufficient even if we build our estimator in other ways that do not make it a Markov chain, and not using a function $Q(.)$ to determine transition probabilities. 

We will now show how to obtain a bound like \eqref{eq:boundonT} when the network instantiations satisfy assumption \ref{asp:networkB}. So consider that the sequence $(A(t))_{t \geq 0}$ is deterministic. Let $(G(t))_{t \geq 0}$ denote the corresponding graph sequence, so that $G(t) = \mathcal{G}(A(t))$. We will now introduce some definitions.

We will say that a sequence $(G(t))_{t_0 \leq t < t_0+m}$ has a sequential path connecting $i$ to $j$ if for $m' \leq m$ there is sequence of edges $((i_k,i_{k+1}))_{0 \leq k < m'}$ with $(i_k,i_{k+1}) \in G(t_0 + k)$, and $i_0 = i$, $i_{m'} = j$. In words, if for the pair $(i,j)$ there is a sequence of edges appearing in succession that connects them. If such a path only exist when we allow for self-loops, \textit{i.e.} a sequential path exist when we allow for $i_k = i_{k+1}$ for some $k$, then we say the sequence has a sequential path with self-loops. 

Consider a sequence $(G(t))_{t_0 \leq t < t_0+m}$. We say it is sequentially connected if for every pair of nodes $(i,j)$, the sequence has a sequential path connecting them. We say the sequence is sequentially connected with self-loops if for every pair $(i,j)$ it has a sequential path when self-loops are allowed. 

We state the following lemma, which relates the definitions just introduced and 
the bound \eqref{eq:boundonT}. 
\begin{lemma}
\label{lem:exponentiallemmanetworkB}
Suppose that $Q(.)$ is chosen so as to satisfy assumption \ref{asp:matrixQ}, and that for all $t_0 \geq 0$ either
\begin{itemize}
\item the sequence $(G(t))_{t_0 \leq t < t_0+m}$ is sequentially connected, or
\item the sequence $(G(t))_{t_0 \leq t < t_0+m}$ is sequentially connected with self-loops, and in addition $Q(.)$ is such that $Q(A)_{ii} \geq \delta$ for all $A \in \mathcal{A}^n$, all $i \in \mathcal{V}$.  
\end{itemize}
Then, we have $\CPR{T^{t_0}_E \leq t + m}{\mathcal{F}_t} \geq \delta^m$, for any $E$ and all $t \geq t_0$. 
\end{lemma}
\begin{proof}
The proof follows in the same way as the proof for lemma \ref{lem:expontentialtaillemma}. In this case, we have $\CPR{\{p(t+1) \in E \} \cup \ldots \cup \{p(t+m) \in E\}}{p(t) = i} \geq \delta^m$, for all $i$,  since we know that at any time window $[t_0, t_0 + m]$ there is a path connecting $i$ to a node in $E$, and the transition probabilities corresponding to choosing this path are all lower bounded by $\delta$.  
\end{proof}

Finally, we can present the following lemma, which relates assumption \ref{asp:networkB} with the sequence of graphs being sequentially connected.
\begin{lemma}
\label{lem:sequentiallyconnected_to}
Suppose that the sequence of graphs is such that assumption \ref{asp:networkB} hold, \textit{i.e.} for some $b$ and all $t_0$, the graph $\bigcup_{t_0 \leq t < t_0 + b} G(t)$ is strongly connected. Then, for each $t_0$, we have that the sequence $(G(t))_{t_0 \leq t < t_0+(n-1)b}$ is sequentially connected with self-loops.
\end{lemma}
\begin{proof}
Consider a generic $t_0$, and let $$\mathcal{E}_k = \bigcup_{t_0 + (k-1)b \leq t < t_0 + kb} \mathcal{E}_{t}, \; k=1,2,\ldots.$$ Consider a generic node $i \in \mathcal{V}$. In the set  $\mathcal{E}_1$ there must be an outward edge from $i$ to at least one node $i_1 \in \mathcal{V}\backslash\{i\}$, as otherwise $\mathcal{E}_1$ is not strongly connected. Thus we can build a path connecting $i$ to $i_1$, possibly using the self-loop in $i$. Let $V(k)$ denote the set of nodes for which we can build a sequential path (with self-loops) starting from $i$ at the end on time window $k$, and then $V(1),V(2),V(3),\ldots \supset \{i,i_1\}$. Now, suppose $V(1) \neq \mathcal{V}$. Then, at time window $k=2$, there must be an outward edge from (at least) one node in $V(1)$ to (at least) one node in $\mathcal{V}\backslash V(1)$, as otherwise the network is not strongly connected at time window $k=2$. We see that as long as $\mathcal{V}\backslash V(k) \neq \emptyset$, we have $|V(k+1)| > |V(k)|$, and thus $|V(n-1)| = n$.   
\end{proof}

Given the previous lemmas, we can establish the following theorems for a network where assumption \ref{asp:networkB} holds. 

\begin{theorem}(Estimator is consistent under assumption \ref{asp:networkB})
\label{thm:theoremconsistencyB}
Let assumptions \ref{asp:noise},\ref{asp:invert} hold, and suppose the network instantiations are such that assumption \ref{asp:networkB} holds. Suppose the function $Q(.)$ is such that assumption \ref{asp:matrixQ} holds, and additionally, that $Q(A)_{ii} \geq \delta$, for all $i$. Then the estimate $\test(t)$ is consistent.  
\end{theorem}
\begin{proof}
Given lemma \ref{lem:consistencyinfinitelly} we just need to show that each agent is visited infinitely often. Given lemma \ref{lem:sequentiallyconnected_to} we have that the results of lemma \ref{lem:exponentiallemmanetworkB} hold with $m = (n-1)\overline{b}$. For a node $i$, consider the sequence of events $E^i_{k} = \{p(km) = i\} \cup \ldots \cup \{p((k+1)m) =i\}$, for $k \geq 0$.  
Define $\mathcal{T}_k = \mathcal{F}_{km}$, and then $E^i_k \in \mathcal{T}_k$. Further, we have for $k \geq 1$, $\CPR{E^i_k}{\mathcal{T}_{k-1}} = \CPR{T^{(k-1)m}_{\{i\}} \leq (k-1)m + m}{\mathcal{F}_{(k-1)m}} > \epsilon$, with $\epsilon$ as given by lemma \ref{lem:exponentiallemmanetworkB}. Hence $\sum_{k} \CPR{E^i_k}{\mathcal{T}_{k-1}} = +\infty$ and by Lévy's extension of the Borel-Cantelli Lemmas (see \cite{williams1991probability} theorem $12.15$) it follows that the sequence of events $E^i_k$ occurs infinitely often, from which follows that node $i$ is visited infinitely often.   
\end{proof}

\begin{theorem}(Estimator is asymptotically optimal under assumption \ref{asp:networkB})
\label{thm:theoremB}
Let assumptions \ref{asp:noise},\ref{asp:invert} hold, and suppose the network instantiations are such that assumption \ref{asp:networkB} holds. Suppose the function $Q(.)$ is such that assumption \ref{asp:matrixQ} holds, and additionally, that $Q(A)_{ii} \geq \delta$, for all $i$. Suppose the sequence $\alpha(t)$ is chosen so that 
\begin{align*}
\lim_{t \rightarrow +\infty} \frac{t \alpha(t)^2}{\E^{c_2 t}} = 0, \;\; \lim_{t \rightarrow +\infty} \frac{t}{\alpha(t)^2} = 0, 
\end{align*}
where $c_2 = \frac{\log(1-\delta^{(n-1)b})}{(n-1)b}$. Then estimator $s(t)$ is asymptotically optimal, so that we have $$\lim_{t \rightarrow +\infty} t \EX{(\test(t) - \theta)(\test(t) - \theta)^\top} = \Sigma_c^{-1}.$$
\end{theorem}
\begin{proof}
We note the conditions of lemma \ref{lem:exponentiallemmanetworkB} hold, with $m = (n-1)b$, and we can state the same result as in lemma \ref{lem:tailforvariables}, with $c_1 = (1 - \delta^{(n-1)b})$ and $c_2 =\frac{\log(1-\delta^{(n-1)b})}{(n-1)b}$. The proof then follows exactly in the same way as for theorem \ref{thm:theorem}. 
\end{proof}

We conclude this section by noting that theorems \ref{thm:theoremconsistencyB} and \ref{thm:theoremB} depended exclusively on the bound \eqref{eq:boundonT}, and the assumptions on the network instantiations and choice of function $Q$ are there to guarantee that the bound holds for some $m,\epsilon$. Thus, the results presented here have the following extension: suppose the sequence of nodes visited by the token $(p(t))_{t \geq 0}$ is such that we can show a bound of the type \eqref{eq:boundonT}. Then the estimate $s(t)$ is consistent and asymptotically optimal given an appropriate choice of $\alpha(t)$. 

\section{Future work}

As could be seen by the simulations, our proposed algorithm can quickly find a low error estimate with relatively few iterations. A weakness of our distributed approach is that it relies on only one estimate being transmitted in the network, and thus if the agent currently carrying this estimate fails, the estimation procedure stops. It would be interesting if we could generalize our methods to a setting where many tokens, carrying different estimates, are traveling in the network, which would add reliability to our algorithm while also taking advantage of the available communications in the network. From a theoretical point of view, we were able to reduce the condition of asymptotic optimality to a single sufficient condition, namely that of \eqref{eq:boundonT}, which lead to an exponential tail on a specific stopping time. However, it is readily seen that the exponential tail is not necessary for the asymptotic optimality of our algorithms, hence it is possible the condition \eqref{eq:boundonT} can be relaxed. Recent work \cite{sahu2018communication} showed how a novel \textit{consensus+innovations} algorithm can achieve asymptotic optimality when the number of communications increase sublinearly with $t$, and it would be interesting to see if we can decrease the communication rate of our algorithm while also guaranteeing asymptotic optimality. Finally, we focused exclusively on asymptotic properties of our estimate, and a different line of study would be the determination of its properties in a finite time window.

% or


\appendices

\section{}
\label{ap:residuals}

In this appendix we show that the terms $R_1(t)$ and $R_2(t)$ as defined in equations \eqref{eq:secondterm} and \eqref{eq:fourthterm} respectively go in expectation to $0$. We will do this this by upper bounding their norms and showing that the expectation of the norm must go to $0$ from which follows that the quantities themselves must go to $0$.

Given a fixed path of the token, the expectation of term \eqref{eq:secondterm} equals
\begin{align*}
& \CEX{R_1(t)}{\mathcal{M}_t}  \\ 
& \quad = t \widetilde{K}(t)^{-1} \left( - \sum_{n=1}^n B_i \theta \IND{i \not \in S(t)} \right) \theta^\top \left( \widetilde{K}(t)^{-1}\Sigma_c - I \right)^\top,
\end{align*}
noting that all the random variables that are present in this last expression are $\mathcal{M}_t$-measurable. 
We have
\begin{align*}
\MoveEqLeft \norm{\CEX{R_1(t)}{\mathcal{M}_t}}  \\
& \leq t \norm{ \widetilde{K}(t)^{-1}} \norm{ \left( - \sum_{i=1}^n B_i \IND{i \not \in S(t)} \right)}  \\
& \qquad\qquad\qquad\qquad\qquad\norm{\left( \widetilde{K}(t)^{-1}\Sigma_c - I \right)} \norm{ \theta}^2 
\end{align*}
and now we can upper bound each of the norms. From the proof of theorem \ref{thm:theorem}, we know that $$\norm{\widetilde{K}(t)^{-1}} \leq \alpha(t)$$ and so $$\norm{\widetilde{K}(t)^{-1}\Sigma_c - I} \leq \norm{\widetilde{K}(t)^{-1}} \norm{\Sigma_c} + 1.$$ We can write
\begin{align*}
\norm{ \left( - \sum_{i=1}^n B_i \IND{i \not \in S(t)} \right)} \leq \norm{\Sigma_c}\IND{S(t) \neq \mathcal{V}}
\end{align*}
where the inequality follows since if some of the nodes have not yet been visited, then the norm of the sum $\sum_{i=1}^n B_i \IND{i \not \in S(t)}$ is upper bounded by the norm of $\Sigma_c$, and if all nodes have already been visited it is $0$. We conclude
\begin{align*}
\norm{\CEX{R_1}{\mathcal{M}_t}} \leq
t \alpha(t) \left( t \norm{\Sigma_c} + 1\right)\norm{\Sigma_c}\norm{\theta}^2 \IND{S_t \neq \mathcal{V}}
\end{align*}
and so
\begin{align}
\label{eq:firstresidual}
& \EX{\norm{\CEX{R_1}{\mathcal{M}_t}}}  \nonumber \\ 
& \qquad\leq t \alpha(t)\left( t \norm{\Sigma_c} + 1\right)\norm{\Sigma_c}\norm{\theta}^2 \PR{S_t \neq \mathcal{V}},
\end{align}
and given the exponential tail on $\PR{S_t \neq \mathcal{V}}$ we find that the quantity in \eqref{eq:firstresidual} goes to $0$.

Finally, for the quantity in \eqref{eq:fourthterm} we have 
\begin{align}
\EX{\norm{R_2(t)}} & = \CEX{\norm{R_2(t)}}{S(t) = \mathcal{V}}\PR{S(t) = \mathcal{V}}  \label{eq:secondresidual_1}\\
& \quad + \CEX{\norm{R_2(t)}}{S(t) \neq \mathcal{V}}\PR{S(t) \neq \mathcal{V}}. \label{eq:secondresidual_2}
\end{align}
Given the exponential tail on $\PR{S(t) \neq \mathcal{V}}$, we can use the bound developed above for $\norm{\widetilde{K}(t)^{-1}\Sigma_c - I}$ to conclude that the term \eqref{eq:secondresidual_2} converges to $0$. Regarding \eqref{eq:secondresidual_1}, we can note that if $S(t) = \mathcal{V}$, all nodes have been visited by time $t$, and then we can write
\begin{align*}
\widetilde{K}(t)^{-1}\Sigma_c - I & = \left( I\alpha(t)^{-1} + \Sigma_c \right)^{-1}\Sigma_c - I \\
&= \left( \Sigma_c^{-1} \alpha(t)^{-1} + I \right)^{-1} - I .
\end{align*}
This matrix is symmetric and thus its $2$-norm equals its spectral radius. Using property \eqref{eq:sumI} we can write its eigenvalues as
\begin{equation*}
\eig\left( \left(\frac{\Sigma_c^{-1}}{\alpha(t)} + I \right)^{-1} - I \right) = \left\{ \frac{-1}{1 + \lambda \alpha(t)} : ~\lambda \in \eig(\Sigma_c) \right\}
\end{equation*}
and thus
\begin{align*}
\norm{K(t)^{-1}\Sigma_c - I} \leq \frac{1}{1+\lambda_{min}(\Sigma_c) \alpha(t)},
\end{align*}
and note that $\lambda_{min}(\Sigma_c) > 0$ since the matrix is positive definite. Finally, we have
\begin{align*}
& \CEX{\norm{R_2(t)}}{S(t) = \mathcal{V}}  \\
& \leq t \norm{\left( \left(\frac{\Sigma_c^{-1}}{\alpha(t)} + I \right)^{-1} - I \right) \theta \theta^\top \left( \left(\frac{\Sigma_c^{-1}}{\alpha(t)} + I \right)^{-1} - I \right)}  \\
& \leq \frac{t}{(1+\lambda_{min}(\Sigma_c)\alpha(t))^2} \norm{\theta}^2
\end{align*}
and this last expression goes to $0$ asymptotically.

%\section*{Acknowledgment}


\ifCLASSOPTIONcaptionsoff
  \newpage
\fi



\bibliographystyle{IEEEtran}
\bibliography{IEEEabrv,./texfiles/papers}
\end{document}